\documentclass[11pt,letterpaper]{article}
\usepackage[utf8]{inputenc}
\usepackage[T1]{fontenc}
\usepackage{microtype}
\usepackage[small]{titlesec}
\usepackage{enumitem}

\usepackage[dvipsnames]{xcolor}
\definecolor{proofparagraphgray}{rgb}{0.282, 0.322, 0.345}
\definecolor{linkblue}{rgb}{0.015, 0.356, 0.745}

\usepackage[colorlinks,allcolors=linkblue,breaklinks=true]{hyperref}

\usepackage{amsmath}
\usepackage{amssymb}
\usepackage{amsthm}

\newtheorem{theorem}{Theorem}
\newtheorem{observation}[theorem]{Observation}
\newtheorem{lemma}[theorem]{Lemma}
\theoremstyle{definition}
\newtheorem{definition}[theorem]{Definition}

\usepackage{tikz}
\usetikzlibrary{quotes,arrows.meta}

\newcommand{\defeq}{:=}

\makeatletter
\newcommand\oeq{\makecircled{\raisebox{.75pt}{\scalebox{0.675}{=}}}}
\newcommand\ole{\makecircled{\raisebox{.75pt}{\scalebox{0.625}{$<$}}}}
\newcommand\ow{\makecircled{\raisebox{.375pt}{\scalebox{0.75}{\sf w}}}}
\newcommand\omm{\makecircled{\raisebox{.375pt}{\scalebox{0.675}{$\lor$}}}}
\newcommand\makecircled[1]
{\mathbin{\ooalign{$\vcenter{\hbox{\scalebox{.75}{$\bigcirc$}}}$\cr\hidewidth$#1$\hidewidth}}}
\makeatother

\def\itE{\tikz[baseline=(X.base), xslant=0.25]\node[inner sep=0pt, xslant=0.25)](X){$\exists$};\!}

\usepackage[normalem]{ulem}
\newcommand{\hl}{\bgroup\markoverwith{\begingroup\color{gray}\raisebox{-.75ex}{.}\endgroup}\ULon}

\title{Non-Boolean OMv: One More Reason to Believe Lower Bounds for Dynamic Problems}

\author{Bingbing Hu \\ \small{UC San Diego} \and Adam Polak \\ \small{Bocconi University}}

\date{}

\newcommand{\proofsubparagraph}[1]{\vspace{-.5em}\paragraph{\color{proofparagraphgray}#1}}

\begin{document}

\maketitle

\begin{abstract}
\noindent Most of the known tight lower bounds for dynamic problems are based on the Online Boolean Matrix-Vector Multiplication (OMv) Hypothesis, which is not as well studied and understood as some more popular hypotheses in fine-grained complexity. It would be desirable to base hardness of dynamic problems on a more believable hypothesis. We propose analogues of the OMv Hypothesis for variants of matrix multiplication that are known to be harder than Boolean product in the offline setting, namely: equality, dominance, min-witness, min-max, and bounded monotone min-plus products. These hypotheses are a priori weaker assumptions than the standard (Boolean) OMv Hypothesis and yet we show that they are actually equivalent to it. This establishes the first such fine-grained equivalence class for dynamic problems.
\end{abstract}

\section{Introduction}

The job of a dynamic algorithm is to keep its output up to date whenever its input undergoes a local change, e.g., maintaining a shortest $s$--$t$ path in a~graph while it undergoes vertex deletions. Ideally each such update should take at most a polylogarithmic time, but at the very least it should be faster than it takes to recompute a solution from scratch. Despite great progress in the field, for many dynamic problems that goal is beyond the reach of current algorithmic techniques. Starting from the seminal paper by Pătraşcu~\cite{Patrascu10}, we often get to explain this hardness by fine-grained conditional lower bounds.

Most of the known \emph{tight} lower bounds for dynamic problems are based on the OMv Hypothesis~\cite{HenzingerKNS15}. This hypothesis is not as widely studied and as well understood as some other hypotheses in fine-grained complexity, such as SETH, 3SUM Hypothesis, and APSP Hypothesis (see, e.g., \cite{Vw18}). It would be more desirable to base hardness of dynamic problems on these more popular (and hence also more believable) assumptions. Unfortunately, the existing lower bounds conditional on them are often not tight for dynamic problems. It seems likely that these hypotheses are not strong enough to explain the complexity of many dynamic problems. We may need to search for a different approach to the following glaring question:
\begin{center}
\emph{Can we have tight lower bounds for dynamic problems\\ based on a hypothesis that is more believable than OMv?}
\end{center}

Recall that the OMv Hypothesis is about \emph{Boolean product}; it asserts that computing the Boolean product of two $n \times n$ matrices requires cubic $n^{3-o(1)}$ time if the second matrix is given column by column in an online fashion. In the static (i.e., non-online) setting, Boolean product is arguably the easiest of the many studied variants of matrix multiplication. Indeed, it can be computed in time $O(n^{\omega})$, where $\omega < 2.372$~\cite{AlmanDWXXZ25} is the (integer) matrix multiplication exponent.\footnote{Moreover, the fastest known ``combinatorial'' algorithm for the Boolean product~\cite{AbboudFKLM24} does not give a similar improvement for the integer product.}

In the static matrix product world, if the $O(n^\omega)$ running time is on the ``fast'' end of the spectrum, then the min-plus product (related to distance computations in graphs) marks the other end: the fastest known algorithm shaves only a subpolynomial factor over the naive cubic running time~\cite{Williams18}, and the APSP Hypothesis from fine-grained complexity essentially says that no $n^{3-o(1)}$-time algorithm is possible~\cite{WilliamsW18}.

There are also numerous variants of matrix multiplication that seem to have an ``intermediate'' hardness on this spectrum. Examples include
min-max product~\cite{VassilevskaWY09,DuanP09},
min-witness product~\cite{CzumajKL07},
equality product (a.k.a.~Hamming product~\cite{MinKZ09}),
dominance product~\cite{Matousek91},
threshold product~\cite{IndykLLP04},
plus-max product~\cite{Vassilevska08},
$\ell_{2p+1}$ product~\cite{LabibUW19},
and many others. The fastest known algorithms for these problems have running times that are functions of the matrix multiplication exponent $\omega$, and they converge to $O(n^{2.5})$ when $\omega=2$. Although it is still an open problem whether this is necessarily the right complexity for all these problems, there are some partial results in the form of tight fine-grained reductions that suggest it might be the case~\cite{LabibUW19, Lincoln0W20, WilliamsX20a}.

\subsection{Our contributions}

The OMv Hypothesis (and the lower bounds it implies) would be a priori more believable if we could replace in its statement the Boolean product with some other product known to be harder in the static world. For instance, we can define in a similar way the Min-Max-OMv problem using the min-max product: Pre-process a matrix $M \in \mathbb{Z}^{n \times n}$, and then answer (one by one, in an online fashion) $n$ queries, each of them asking, given a~vector $v \in \mathbb{Z}^n$, to compute the min-max product of $M$ and $v$, i.e., the vector $u \in \mathbb{Z}^n$ such that
\[u[i] \defeq \min_{k \in [n]}\max\{M[i,k], v[k]\}.\]
Then we can state a corresponding hypothesis, let us call it the Min-Max-OMv Hypothesis, asserting that the Min-Max-OMv problem cannot be solved in truly subcubic time $O(n^{3-\varepsilon})$, for any $\varepsilon > 0$. This of course brings a question:

\begin{center}
\emph{Can we still give tight reductions from Min-Max-OMv to those dynamic problems for which there are known reductions from (Boolean-)OMv?}
\end{center}

It turns out, yes, we can! Somewhat surprisingly\footnote{This is perhaps less surprising to those readers who are familiar with how the subcubic algebraic algorithms for the intermediate problems work; see Section~\ref{sec:techoverview}.
}, we can even give a~tight reduction from Min-Max-OMv to Boolean-OMv. This shows that the Min-Max-OMv Hypothesis and the standard (Boolean-)OMv Hypothesis are actually equivalent. Moreover, the min-max product is not a unique example of this phenomenon. We show more equivalent hypotheses based on several matrix products, which are harder than the Boolean product in the static setting (see Section~\ref{sec:preliminaries} for the formal definitions).

\begin{theorem}
\label{thm:main}
    The following problems either all have truly subcubic algorithms or none of them do:
    \begin{itemize}[nosep,left=0pt]
      \item Boolean-OMv;
        \hfill {\color{gray}$(\exists_{k}\,M[i,k] \land v[k])$}
      \item \itE{}Equality-OMv;
        \hfill {\color{gray}$(\exists_{k}\,M[i,k] = v[k])$}
      \item \itE{}Dominance-OMv; 
        \hfill {\color{gray}$(\exists_{k}\,M[i,k] \leqslant v[k])$}
      \item Min-Witness-OMv;
        \hfill {\color{gray}$(\min\,\{k \mid M[i,k] \land v[k]\})$}
      \item Min-Max-OMv;
        \hfill {\color{gray}$(\min_k \max\{M[i,k], v[k]\})$}
      \item Bounded Monotone Min-Plus-OMv.
        \hfill {\color{gray}$(\min_k M[i,k] + v[k])$}
    \end{itemize}
For the Bounded Monotone Min-Plus-OMv problem the implied algorithm is randomized.
\end{theorem}

\pagebreak  

This conglomeration of equivalent problems can be interpreted as making the OMv Hypothesis more believable, and the conditional lower bounds based on it stronger. We recall two analogous conglomerations: the NP-complete problems and the problems equivalent to the All-Pairs Shortest Paths (APSP) problem under subcubic reductions. One of the reasons behind the great success of the theory of NP-completeness is its structural simplicity: many natural problems are NP-complete, and solving any of them efficiently would solve all of them efficiently, so they are all hard for the same underlying reason. For all the multi-faceted NP-complete problems, researchers from different areas have not managed to find a single efficient algorithm, so it seems very plausible that no such algorithm exists. The fine-grained complexity theory at large does not enjoy a similar simplicity: there are multiple independent hardness assumptions, and the reductions often go only in one way, establishing \emph{hardness} but not \emph{equivalence}. A notable exception\footnote{Another such exception is a class of vector problems equivalent to the Orthogonal Vectors problem~\cite{ChenW19}.} is the APSP problem, which is conjectured to require cubic time and there are many other problems equivalent to it via subcubic reductions~\cite{WilliamsW18}. No truly subcubic algorithms have been found so far for any of these problems, which strengthens the APSP Hypothesis. Our Theorem~\ref{thm:main} establishes another such class of problems equivalent under fine-grained subcubic reductions.

\paragraph*{Subcubic algorithms in the cell-probe model.} We remark that the above equivalences also hold in the cell-probe model, where the running time is measured solely by the number of memory cell probes, and computation is free of charge. To be more specific, from our proofs it follows that each of the above problems can be solved by solving $t^{O(1)}$ instances of any other of those problems and performing some additional computation that runs in time $O(n^3/t)$, in the word RAM model, where $t$ is a parameter to be chosen. Since any algorithm that runs sequentially in time $O(n^{3-\varepsilon})$ in the word RAM model can access only a truly subcubic number of memory cells, the subcubic equivalence in the cell-probe model follows. Larsen and Williams~\cite{LarsenW17} showed that the OMv Hypothesis is actually \emph{false} in the cell-probe model: Every query can be computed with $O(n^{7/4}w^{-1/2})$ probes, where $w$ is the word size. This means that all the other problems listed above have truly subcubic algorithms in the cell-probe model as well.

\subsection{Technical overview}
\label{sec:techoverview}

We prove Theorem~\ref{thm:main} by a series of fine-grained reductions, depicted in Figure~\ref{fig:reductions}. The reductions are inspired by known subcubic algebraic algorithms for the corresponding (static) matrix product problems. Such algorithms have running times that are functions of the (integer) matrix multiplication exponent~$\omega$. For example, the fastest known (static) min-max product algorithm~\cite{DuanP09} runs in time $O(n^{(3+\omega)/2 + o(1)})$, which is subcubic for any value of $\omega < 3$. In other words, this algorithm is an implicit subcubic reduction from the (static) min-max product to the (static, integer) matrix multiplication problem. The same is true about all the other intermediate problems we consider in this paper. Our technical contribution is adapting these implicit reductions so that (1) they work in the online setting, and (2) they require only Boolean (and not integer) product.

\paragraph{Detour into static combinatorial algorithms.}
Before we say a few words about each of the reductions in the paper, let us take a detour from online problems so that we can explicitly state the observation that we already hinted at in the previous paragraph.

\begin{observation}
The static variants of all the matrix products listed in Theorem~\ref{thm:main} are equivalent to each other under subcubic reductions, and the reductions are ``combinatorial''. That is, if any of the static problems admits a subcubic ``combinatorial'' algorithm, then all of them do.
\end{observation}

Why is this the case? In one direction, all those non-Boolean products can encode the Boolean product as a special case. In the other direction, all those non-Boolean products admit strongly subcubic algorithms with running times of the form $O(n^{f(\omega)})$ for a function $f$ that happens to satisfy $f(x) < 3$ for every $x < 3$. Therefore, if we could replace in such an algorithm every use of fast algebraic (integer) matrix multiplication with a hypothesized $O(n^{3-\varepsilon})$-time ``combinatorial'' Boolean matrix multiplication algorithm, we would get a ``combinatorial'' algorithm for the corresponding non-Boolean product with the running time exponent $f(3-\varepsilon) < 3$. Such a replacement is possible whenever the non-Boolean algorithm does not use (or can be modified not to use) the full counting power of integer matrix multiplication but only checks which output entries are nonzero. This is explicitly the case for some of the non-Boolean products we consider (e.g., the min-witness product~\cite{CzumajKL07}). For some others, e.g., the bounded monotone min-plus product~\cite{WilliamsX20b,Gu0WX21}, such a modification is nontrivial, but still possible.

While this observation is simple and in hindsight not really surprising, we are not aware of it being folklore in the field of fine-grained complexity, so we state it for completeness.

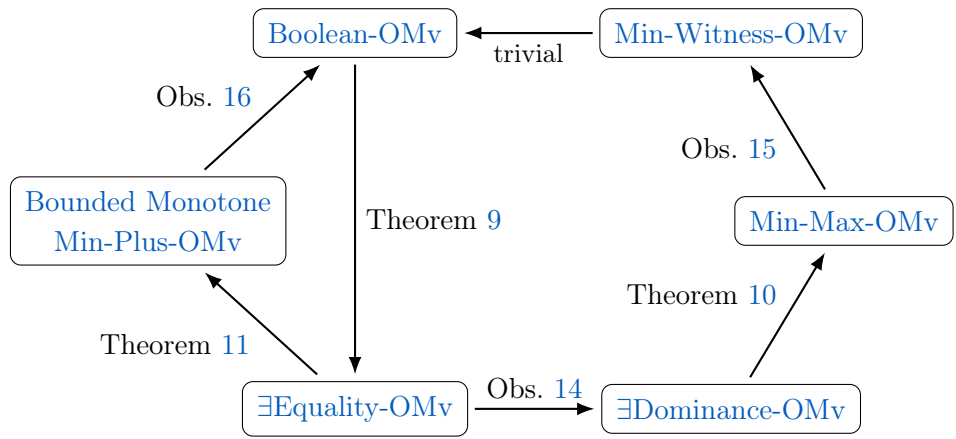
\begin{figure}
\centering
\begin{tikzpicture}[problem/.style={draw,rounded corners,inner sep=.5em,outer sep=.25em}, every edge/.style={draw,thick,-Latex}]

\node [problem] (bool) at (0,5) {\hyperref[def:bool]{Boolean-OMv}};
\node [problem] (eq)   at (0,0) {\hyperref[def:eq]{$\exists$Equality-OMv}};
\node [problem] (dom)  at (5,0) {\hyperref[def:dom]{$\exists$Dominance-OMv}};
\node [problem] (mwit) at (5,5) {\hyperref[def:wit]{Min-Witness-OMv}};
\node [problem] (mmax) at (6.5,2.5) {\hyperref[def:minmax]{Min-Max-OMv}};
\node [problem] (bmmp) at (-2.75,2.5) [align=center] {\hyperref[def:minplus]{Bounded Monotone} \\ \hyperref[def:minplus]{Min-Plus-OMv}};

\draw (bool) edge["Theorem~\ref{thm:eq}"] (eq);
\draw (eq)   edge["Obs.~\ref{obs:dom}"] (dom);
\draw (dom)  edge["Theorem~\ref{thm:minmax}"] (mmax);
\draw (mmax) edge["Obs.~\ref{obs:wit}"] (mwit);
\draw (mwit) edge["\small trivial"] (bool);
\draw (eq)   edge["Theorem~\ref{thm:minplus}"] (bmmp);
\draw (bmmp) edge["Obs.~\ref{obs:bool}"] (bool);
\end{tikzpicture}
\caption{Fine-grained reductions that together prove Theorem~\ref{thm:main}. An arrow from problem A to problem B means that a subcubic algorithm for A implies a subcubic algorithm for B.}
\label{fig:reductions}
\end{figure}

\paragraph{Subcubic reductions and where to find them.}
In Section~\ref{sec:eq} we show a subcubic reduction from $\exists$Equality-OMv to Boolean-OMv (Theorem~\ref{thm:eq}), which can be seen as an adaptation to the online setting of the sparse matrix multiplication algorithm of Yuster and Zwick~\cite{YusterZ05}. Specifically, our algorithm for $\exists$Equality-OMv handles the $n^{\varepsilon/2}$ most frequently appearing values with the help of $n^{\varepsilon/2}$ Boolean-OMv queries (about indicator vectors of those values). The remaining less frequent values can then be handled with a brute-force approach within the target running time bound.

In Section~\ref{sec:minmax} we show how a subcubic algorithm for  $\exists$Dominance-OMv would yield a subcubic algorithm for Min-Max-OMv (Theorem~\ref{thm:minmax}). The proof is inspired by the known (static) min-max product algorithms~\cite{VassilevskaWY09,DuanP09}, but it is at the same time simpler, because we do not have to optimize the dependence on $\varepsilon$ in the running time of the resulting algorithm. In short, the algorithm splits each row of $M$ into $n^{\varepsilon/2}$ buckets based on the values of the entries, and makes $n^{\varepsilon/2}$ many $\exists$Dominance-OMv queries in order to establish for each output entry to which bucket of the corresponding row it belongs; then, the algorithm exhaustively searches through the entire bucket.

In Section~\ref{sec:minplus} we show a reduction from Bounded Monotone Min-Plus-OMv to $\exists$Equality-OMv. On a high level it follows some of the previous (static) algorithms for the bounded monotone min-plus product~\cite{WilliamsX20b,Gu0WX21}. However, it also gives a fresh perspective on the problem, because those previously known algorithms use a generalization of the min-witness and bounded (non-monotone) min-plus products (see~\cite[Theorem~1.2]{WilliamsX20b}), while ours deviates from this approach by using the equality product.

In Section~\ref{sec:remaining} we show the remaining reductions (Observations~\ref{obs:dom}, \ref{obs:wit}, \ref{obs:bool}). Each of them is either very simple or follows easily from folklore arguments.

\subsection{Related work}

Bringmann et al.~\cite{BringmannGKL24} take a different approach at strengthening the OMv Hypothesis. They propose a hypothesis about the complexity of determining if a (nondeterministic) finite automaton accepts a word, and show that this hypothesis implies the OMv Hypothesis. While their new hypothesis is not as well supported as the three main fine-grained complexity hypotheses, it is remarkable that it is a statement about a \emph{static} problem implying a tight lower bound for an online problem.

In a very recent work, Liu~\cite{Liu24} shows that OMv is equivalent to the online problem of maintaining a $(1+\epsilon)$-approximate vertex cover in a fully dynamic bipartite graph.

To the best of our knowledge, the only other work that considers a variant of OMv for a non-Boolean product is by Chen et al.~\cite{ChenDGWXY18}. They use an assumption that the Min-Plus-OMv requires cubic time in order to separate partially retroactive from fully retroactive data structures. We note that this assumption seems too strong to be equivalent to the OMv Hypothesis. In particular, any ``too simple'' reduction from Min-Plus-OMv to Boolean-OMv would morally translate to a subcubic algorithm for the (static) Min-Plus Product problem, refuting the APSP Hypothesis.

\subsection{Open problems}

In this paper we manage to reduce to Boolean-OMv from OMv variants that \emph{do not involve counting}. We leave it open whether a subcubic algorithm for Boolean-OMv would imply subcubic OMv algorithms for, e.g., the counting variants of the equality and dominance products (i.e., $u[i] := \#\{k \mid M[i,k] = v[k]\}$, and $u[i] := \#\{k \mid M[i,k] \leqslant v[k]\}$, respectively), or at least for the standard integer product ($u[i] := \sum_k M[i,k] \cdot v[k]$). 

These open problems relate to the general quest for fine-grained \emph{counting-to-decision} reductions. Chan, Vassilevska Williams, and Xu~\cite{ChanWX23} gave such reductions for the Min-Plus Product, Exact Triangle, and 3SUM problems. Somewhat ironically, their reductions crucially rely on a fast algebraic algorithm for (static) integer matrix multiplication, so it seems unlikely that their techniques could be used to resolve the above open problems, which are about online problems.

\section{Preliminaries}
\label{sec:preliminaries}

\subsection{Notation}

We use $[n] := \{1, 2, \ldots, n\}$.

\subsection{Problems}

In this section we formally define all the problems that appear in Theorem~\ref{thm:main}. Since the definitions are similar to each other, we \hl{underline} the differences between them.

\begin{definition}[\hl{Boolean}-OMv]\label{def:bool}
We are first given for preprocessing a \hl{Boolean} matrix $M \in \hl{\{0, 1\}^{n \times n}}$, and then we need to answer $n$ queries: In the $j$-th query, we are given a column vector $v_j \in \hl{\{0, 1\}^n}$, and we have to compute the \hl{Boolean product $Mv_j \in \{0,1\}^n$} defined by
\[\hl{(Mv_j)[i] \defeq \begin{cases} 1, & \text{if}\ \exists_{k \in [n]} \, M[i,k]=1 \land v_j[k]=1, \\ 0, & \text{otherwise}.\end{cases}}\]
We need to answer queries one by one in an online fashion, i.e., we have to output $Mv_j$ before we can receive $v_{j+1}$.
\end{definition}

\begin{definition}[\hl{$\exists$Equality}-OMv]\label{def:eq}
We are first given for preprocessing an \hl{integer} matrix $M \in \hl{\mathbb{Z}^{n \times n}}$, and then we need to answer $n$ queries: In the $j$-th query, we are given a column vector $v_j \in \hl{\mathbb{Z}^n}$, and we have to compute the \hl{$\exists$equality product $M \oeq v_j \in \{0,1\}^n$} defined by
\[\hl{(M \oeq v_j)[i] \defeq \begin{cases} 1, & \text{if}\ \exists_{k \in [n]} \, M[i,k] = v_j[k], \\ 0, & \text{otherwise}.\end{cases}}\]
We need to answer queries one by one in an online fashion, i.e., we have to output $M \oeq v_j$ before we receive $v_{j+1}$.
\end{definition}

\begin{definition}[\hl{$\exists$Dominance}-OMv]\label{def:dom}
We are first given for preprocessing an \hl{integer} matrix $M \in \hl{\mathbb{Z}^{n \times n}}$, and then we need to answer $n$ queries: In the $j$-th query, we are given a column vector $v_j \in \hl{\mathbb{Z}^n}$, and we have to compute the \hl{$\exists$dominance product $M \ole v_j \in \{0,1\}^n$} defined by
\[\hl{(M \ole v_j)[i] \defeq \begin{cases} 1, & \text{if}\ \exists_{k \in [n]} \, M[i,k] \leqslant v_j[k], \\ 0, & \text{otherwise}.\end{cases}}\]
We need to answer queries one by one in an online fashion, i.e., we have to output $M \ole v_j$ before we receive $v_{j+1}$.
\end{definition}

\begin{definition}[\hl{Min-Witness}-OMv]\label{def:wit}
We are first given for preprocessing a~\hl{Boolean} matrix $M \in \hl{\{0, 1\}^{n \times n}}$, and then we need to answer $n$ queries: In the $j$-th query, we are given a column vector $v_j \in \hl{\{0, 1\}^n}$, and we have to compute the \hl{min-witness product $M \ow v_j \in ([n]\cup\{\infty\})^n$} defined by
\[\hl{(M \ow v_j)[i] \defeq \min (\{k \in [n] \mid M[i,k]=1 \land v_j[k]=1\} \cup \{\infty\}).}\]
We need to answer queries one by one in an online fashion, i.e., we have to output $M \ow v_j$ before we can receive $v_{j+1}$.
\end{definition}

\begin{definition}[\hl{Min-Max}-OMv]\label{def:minmax}
We are first given for preprocessing an \hl{integer} matrix $M \in \hl{\mathbb{Z}^{n \times n}}$, and then we need to answer $n$ queries: In the $j$-th query, we are given a column vector $v_j \in \hl{\mathbb{Z}^n}$, and we have to compute the \hl{min-max product $M \omm v_j \in \mathbb{Z}^n$} defined by
\[\hl{(M \omm v_j)[i] \defeq \min_{k \in [n]}\max\{M[i,k], v_j[k]\}.}\]
We need to answer queries one by one in an online fashion, i.e., we have to output $M \omm v_j$ before we receive $v_{j+1}$.
\end{definition}

\begin{definition}[\hl{Bounded Monotone Min-Plus}-OMv]\label{def:minplus}
We are first given for preprocessing an \hl{integer} matrix $M \in \hl{[n]^{n \times n}}$, and then we need to answer $n$~queries: In the $j$-th query, we are given a column vector $v_j \in \hl{[n]^n}$, and we have to compute the \hl{min-plus product $M \oplus v_j \in \mathbb{Z}^n$} defined by
\[\hl{(M \oplus v_j)[i] \defeq \min_{k \in [n]} (M[i,k] + v_j[k]).}\]
We need to answer queries one by one in an online fashion, i.e., we have to output $M \oplus v_j$ before we receive $v_{j+1}$.
\hl{We are guaranteed that at least one of the following conditions holds:}
\begin{itemize}[nosep,left=0pt]
    \item \hl{each row of $M$ is nondecreasing, i.e., $M[i,k] \leqslant M[i,k+1]$;}
    \item \hl{each column of $M$ is nondecreasing, i.e., $M[i,k] \leqslant M[i+1,k]$;}
    \item \hl{each $v_j$ is nondecreasing, i.e., $v_j[k] \leqslant v_j[k+1]$;}
    \item \hl{for every $k$, $v_j[k]$ is a nondecreasing function of $j$, i.e., $v_j[k] \leqslant v_{j+1}[k]$;}
\end{itemize}
\hl{or it holds with ``nondecreasing'' replaced by ``nonincreasing''.}
\end{definition}

\subsection{Hypotheses} Each of the problems defined above admits a naive cubic time algorithm, and for each of them we can conjecture that it is optimal up to subpolynomial factors.

\begin{definition}[*-OMv Hypotheses]\label{def:hypotheses}
For $x \in \{$Boolean, $\exists$Equality, $\exists$Dominance, Min-Witness, Min-Max, Bounded Monotone Min-Plus$\}$, the $x$-OMv Hypothesis is the statement that there is no algorithm for the $x$-OMv problem running in time $O(n^{3-\varepsilon})$, for any $\varepsilon > 0$.
\end{definition}
In other words, Theorem~\ref{thm:main} says that all the hypotheses stated in Definition~\ref{def:hypotheses} are equivalent.

\section{Reduction from \texorpdfstring{$\boldsymbol{\exists}$}{∃}Equality-OMv to Boolean-OMv}
\label{sec:eq}

In this section we give a subcubic reduction that can be seen as an easy adaptation to the online setting of the sparse matrix multiplication algorithm of Yuster and Zwick~\cite{YusterZ05}.

\begin{theorem}\label{thm:eq}
If Boolean-OMv can be solved in time $O(n^{3-\varepsilon})$, for some $\varepsilon > 0$, then \itE{}Equality-OMv can be solved in time $O(n^{3-(\varepsilon/2)})$.
\end{theorem}

\begin{proof}
Recall that $M$ denotes the input matrix given for preprocessing in the $\exists$Equality-OMv problem. Let $t := \lceil n^{\varepsilon / 2} \rceil$ be a parameter. For every $k \in [n]$ and every $\ell \in [t]$, let $f_k^{(\ell)}$ be the $\ell$-th most frequent value appearing in the $k$-th column of matrix $M$ (if there are less than $\ell$ distinct values in the column, let $f_k^{(\ell)}$ be some other arbitrary integer). Note that for any value $x$ not in $\{f_k^{(1)}, f_k^{(2)}, \ldots, f_k^{(t)}\}$, $x$ appears in the $k$-th column of $M$ at most $n / t$ times; we call such values \emph{rare}. In the preprocessing phase, the algorithm prepares $t$~Boolean matrices $M^{(1)}, M^{(2)}, \ldots, M^{(t)} \in \{0,1\}^{n \times n}$ defined as follows:
\[M^{(\ell)}[i,k] \defeq \begin{cases} 1, & \text{if}\ M[i,k] = f_k^{(\ell)}, \\ 0, & \text{otherwise}.\end{cases}\]
Then, it instantiates the hypothesized Boolean-OMv algorithm for each of these matrices separately. Finally, for each column of $M$, the algorithm prepares a dictionary mapping each rare value in that column to a list of indices under which that value appears in the column. This ends the preprocessing phase.

Upon receiving a query $v \in \mathbb{Z}^n$, the algorithm first initializes the output vector to all zeros. Then, for every $\ell = 1, \ldots, t$, it creates the vector $v^{(l)}$ defined by
\[v^{(\ell)}[k] \defeq \begin{cases} 1, & \text{if}\ v[k] = f_k^{(\ell)}, \\ 0, & \text{otherwise},\end{cases}\]
and computes the Boolean product $M^{(\ell)}v^{(\ell)}$, using the $\ell$-th instantiation of the hypothesized Boolean-OMv algorithm. Each such product gets then element-wise OR-ed to the output vector. Finally, for every $k = 1, \ldots, n$, if $v[k]$ is a rare value in the $k$-th column of matrix $M$, the algorithm goes through the list of all indices $i$ such that $M[i][k]=v[k]$ (recall that there are at most $n/t$ of them) and for each of them sets the corresponding $i$-th entry of the output vector to~$1$.

It is easy to see that whenever the algorithm sets an output entry to $1$, it is because of some pair of entries $M[i][k]$ and $v[k]$ that have the same value. Conversely, if some pair of entries $M[i][k]$ and $v[k]$ have the same value, then either it is a frequent value and some $M^{(\ell)}v^{(\ell)}$ contributes a $1$, or it is a~rare value and gets manually matched.

Let us analyze the running of our $\exists$Equality-OMv algorithm. There are $t$ instantiations of the hypothesized Boolean-OMv algorithm, which require $O(tn^{3-\varepsilon})$ time in total. Then, going through all rare values takes at most $O(n^2/t)$ time per $v_j$, and thus $O(n^3 / t)$ time for all $n$ queries. This adds up to total time $O(tn^{3-\varepsilon} + n^3 / t)$. By choosing $t := \lceil n^{\varepsilon / 2} \rceil$ we get the claimed running time $O(n^{3-(\varepsilon/2)})$.
\end{proof}

\section{Reduction from Min-Max-OMv to \texorpdfstring{$\boldsymbol{\exists}$}{∃}Dominance-OMv}
\label{sec:minmax}

In this we show how a subcubic algorithm for $\exists$Dominance-OMv would yield a subcubic algorithm for Min-Max-OMv. Our proof is inspired by the static min-max product algorithms~\cite{VassilevskaWY09,DuanP09}, but it is at the same time simpler, because we do not optimize the dependence on $\varepsilon$ in the running time of the resulting algorithm.

\begin{theorem}\label{thm:minmax}
If \itE{}Dominance-OMv can be solved in time $O(n^{3-\varepsilon})$, for some $\varepsilon > 0$, then Min-Max-OMv can be solved in time $O(n^{3-(\varepsilon/2)})$.
\end{theorem}

\begin{proof} 
Let $t \defeq \lceil n^{\varepsilon / 2} \rceil$ be a parameter. For every $i \in [n]$, let $R_i$ be the sorted $i$-th row of the input matrix $M$. Consider partitioning each $R_i$ into $t$ \emph{buckets} of consecutive elements, with at most $\lceil n / t \rceil$ elements per bucket. For every $\ell \in [t]$, let $M^{(\ell)} \in (\mathbb{Z}\cup\{\infty\})^{n \times n}$ be the matrix defined as follows:
\[M^{(\ell)}[i,k] \defeq \begin{cases} -M[i,k], & \text{if $M[i,k]$ lands in the $\ell$-th bucket of $R_i$},\\ \infty, & \text{otherwise}.\end{cases}\]
Note that each row of $M^{(\ell)}$ contains $\Theta(n/t)$ finite entries.\footnote{If there are multiple entries with the same value, they may land in different buckets.}

In the preprocessing phase, the algorithm instantiates the hypothesised $\exists$Dominance-OMv algorithm for each of the matrices $M^{(1)}, M^{(2)}, \ldots, M^{(\ell)}$, and also for the matrix $M$.\footnote{Formally, the $\exists$Dominance-OMv algorithm may not accept infinite entries in the input, but we can replace each $\infty$ with $W+2$, where $W$ denotes the largest absolute value of any entry in $M$, and each entry greater than $W$ in any query vector with $W+1$.}

Upon receiving a query $v \in \mathbb{Z}^n$, the algorithm proceeds to compute the product $M \omm v$ in two steps. First, for every $i \in [n]$, it computes the minimum $M[i,k]$ such that $M[i,k] \geqslant v[k]$, and stores the results in a column vector $u$. Second, for every $i \in [n]$, it computes the minimum $v[k]$ such that $v[k] \geqslant M[i,k]$, and stores the results in a column vector $w$. At the very end the algorithm computes $(M \omm v)[i] = \min \{ u[i], w[i] \}$, for every $i \in [n]$.

In order to compute $u$, the algorithm first asks for the dominance products $M^{(\ell)} \ole (-v)$, for all $\ell \in [t]$. Then, for each $i = 1, \ldots, n$, the algorithm finds the smallest $\ell$ such that $(M^{(\ell)} \ole (-v))[i] = 1$, which corresponds to finding the first bucket in $R_i$ containing an element \emph{greater}\footnote{This is because the entries in $M^{(\ell)}$ and $-v$ are negated.} than or equal to the corresponding element in $v$. Hence, the algorithm can scan the elements in this bucket and pick the smallest one that is larger than or equal to the corresponding element in $v$; this element is then stored in $u[i]$.

Let us analyze the cost of computing $u$'s over the span of $n$ queries. The $t$~dominance products require time $O(tn^{3-\varepsilon})$ in total. On top of that, for each of the $n$ queries and for each of the $n$ output coordinates, the algorithm scans one bucket of size $\Theta(n/t)$, which takes time $O(n^3/t)$ in total. All together, the algorithm spends time $O(tn^{3-\varepsilon} + n^3/t) = O(n^{3-(\varepsilon / 2)})$ on computing $u$'s.

Next, it is almost symmetric to calculate $w$. The algorithm sorts the entries of $v$ into an ordered list $S$, and partitions $S$ into $t$ buckets, with at most $\lceil n/t \rceil$ elements per bucket. For each bucket $\ell \in [t]$, the algorithm computes the dominance product $M \ole v^{(\ell)}$, where $v^{(\ell)} \in (\mathbb{Z}\cup\{-\infty\})^n$ is the column vector such that
\[v^{(\ell)}[k] = \begin{cases} v[k], & \text{if $v[k]$ lands in the $\ell$-th bucket of $S$}, \\ -\infty, & \text{otherwise}.
\end{cases}\]
Then, for each $i = 1, \ldots, n$, the algorithm looks for the smallest $\ell$ such that $(M \ole v^{(\ell)})[i] = 1$, and scans the elements in the $\ell$-th bucket looking for the smallest $v[k]$ that is greater than or equal to the corresponding $M[i,k]$. By the same argument as before, computing all $w$'s takes time $O(n^{3 - (\varepsilon / 2)})$.
\end{proof}

\section{Reduction from Bounded Monotone Min-Plus-OMv to \texorpdfstring{$\boldsymbol{\exists}$}{∃}Equality-OMv}
\label{sec:minplus}

In this section we show a reduction from Bounded Monotone Min-Plus-OMv to $\exists$Equality-OMv. We follow the high-level approach of some of the previous static algorithms for the bounded monotone min-plus product~\cite{WilliamsX20b,Gu0WX21}. However, we deviate from that approach by using the equality product where the known algorithms use a generalization of the min-witness and bounded (non-monotone) min-plus products (see~\cite[Theorem~1.2]{WilliamsX20b}).

\begin{theorem}\label{thm:minplus}
If \itE{}Equality-OMv can be solved in time $O(n^{3-\varepsilon})$, for some $\varepsilon > 0$, then Bounded Monotone Min-Plus-OMv can be solved in time $O(n^{3 - (\varepsilon / 3)} \log n)$ by a randomized algorithm that succeeds with probability\footnote{Note that the success probability can be amplified to $1-1/\operatorname{poly}(n)$ by running in parallel a constant number of copies of the algorithm and taking the majority vote.} at least $1-1/n$.
\end{theorem}

Before we present the algorithm itself let us introduce some notation and prove some preliminary facts. Let $\Delta \defeq \lceil n^{\varepsilon / 3} \rceil$ be a parameter. For a fixed query vector $v \in \mathbb{Z}^n$, let
\[u \defeq M \oplus v, \quad \widehat{M} \defeq \lfloor M / \Delta \rfloor, \quad \widehat{v} \defeq \lfloor v / \Delta \rfloor, \quad \text{and} \quad \widehat{u} \defeq \widehat{M} \oplus \widehat{v}.\]
Be mindful that it is \emph{not} necessarily the case that $\widehat{u} = \lfloor u / \Delta \rfloor$. Finally, for every $i \in [n]$, let us define the set of \emph{candidates for} $u[i]$ to be
\[C_i \defeq \bigl\{ k \in [n] \bigm\vert \widehat{M}[i,k] + \widehat{v}[k] \in \{\widehat{u}[i], \widehat{u}[i]+1\} \bigr\}.\]

\begin{lemma}\label{lem:sufficientCi}
It suffices to check only $k \in C_i$ in order to compute $u[i]$, i.e.,
\[\min_{k\in[n]} M[i,k]+v[k] = \min_{k \in C_i} M[i,k]+v[k].\]
\end{lemma}

\begin{proof}
First, for any pair $(i, j) \in [n] \times [n]$, due to rounding down we have
\[M[i,j] + v[j] - \Delta \cdot (\widehat{M}[i, j] + \widehat{v}[j]) \in [0, 2\Delta).\]
Now, suppose that $k$ is a witness for $u[i]$, and $l$ is a witness for $\widehat{u}[i]$, i.e., $M[i,k] + v[k] = u[i]$, and $\widehat{M}[i,l] + \widehat{v}[l] = \widehat{u}[i]$. We derive that
\begin{align*}
    \Delta \cdot \widehat{u}[i] + 2\Delta &= \Delta \cdot (\widehat{M}[i,l] + \widehat{v}[l]) + 2\Delta \\
    &> M[i,l] + v[l] \\
    &\geqslant M[i,k] + v[k] \\
    & \geqslant \Delta \cdot (\widehat{M}[i,k] + \widehat{v}[k]).
\end{align*}
Therefore, we have $\widehat{M}[i,k] + \widehat{v}[k] < \widehat{u}[i] + 2$. Since the matrix entries all take integer values, we have that if $k \in [n]$ is a witness for $u[i]$, then it must satisfy that $\widehat{M}[i,k] + \widehat{v}[k] \in \{\widehat{u}[i], \widehat{u}[i]+1\}$, i.e., $k \in C_i$.
\end{proof}

Now we argue that small sets of candidates can be enumerated efficiently.

\begin{lemma}\label{lem:listCi}
For a fixed query vector $v \in \mathbb{Z}^n$, there is an algorithm that
runs in time $O(n^2\log n / \Delta)$ and lists all elements of all sets $C_i$ such that $|C_i| \leqslant n / \Delta$.
In the case that $v_j[k]$ is a (weakly) monotone function of $j$ (i.e., the 4-th case in Definition~\ref{def:minplus}) this running time is amortized over $n$ query vectors.
\end{lemma}

\begin{proof}
We consider four cases, based on the direction of the monotonicity:

\proofsubparagraph{(1) Each column of $\boldsymbol{M}$ is monotone.}
In this case also the columns of $\widehat{M}$ are monotone, and their entries are bounded by $\lfloor n/\Delta \rfloor$. The algorithm uses a self-balancing binary search tree (BST) to maintain, while $i$ iterates from $1$ to $n$, the set of pairs
\[\big\{(\widehat{M}[i,k] + \widehat{v}[k], k) \bigm\vert k \in [n]\big\}.\]
Computing $\widehat{u}[i]$ is the standard tree operation of querying for the minimum. Moreover, the BST can report the number of elements smaller than a certain value in time $O(\log n)$, and enumerate them in time proportional to that number. This allows the algorithm to determine the size of $C_i$ quickly, and enumerate it if it is small. As $i$ iterates from $1$ to $n$, the algorithm only needs to update the elements where there is an increase (or decrease) from $\widehat{M}[i,k]$ to $\widehat{M}[i+1,k]$. In each column of $M$ there are at most $n/\Delta$ such changes, thanks to the boundedness and monotonicity.\footnote{The locations of all such changes can be, e.g., precomputed before processing any query, or computed during each query using binary search in time $O(\log n)$ per each location.} Therefore the total number of updates over the $n$ iterations is at most $n^2/\Delta$, and each update takes time $O(\log n)$. The time spent on listing elements of $C_i$ (for all $i$) is $O(n \log n + n^2/\Delta)$.

\proofsubparagraph{(2) For each $\boldsymbol{k}$, $\boldsymbol{v_j[k]}$ is a monotone function of $\boldsymbol{j}$.}
This case is very similar to the previous one. The algorithm maintains (over the span of $n$~queries) a separate BST for each $i$, and uses it to compute $(\widehat{M}\oplus\widehat{v}_j)[i]$ for all $j$'s. When there is an increase (or decrease) from $\widehat{v}_j[k]$ to $\widehat{v}_{j+1}[k]$, the algorithm has to update an element in all $n$ trees, but this happens at most $n/\Delta$ times for each $k$, so $n^2/\Delta$ times for all $k$'s. Hence, the total time spent on such updates over the course of $n$ queries is $O(n^3\log n/\Delta)$, and the amortized time per query is $O(n^2\log n/\Delta)$.

\proofsubparagraph{(3) Each row of $\boldsymbol{M}$ is monotone.}
Due to the monotonicity, we can think of the $i$-th row of $\widehat{M}$, for each $i=1,\ldots,n$, as consisting of $\lfloor n / \Delta \rfloor + 1$ contiguous \emph{blocks} $K^{(0)}_i, K^{(1)}_i, \dots, K^{(\lfloor n / \Delta \rfloor)}_i \subseteq [n]$ of identical entries, i.e., $\forall \raisebox{-.25em}{$\scriptstyle k \in K^{(x)}_i$} \:  M[i,k]=x$. Upon receiving a query vector $v$, the algorithm builds (in linear time) a range minimum query (RMQ) data structure (see, e.g., \cite{BenderF00}) in order to compute in constant time the minimum entry of $\widehat{v}$ in each of the $O(n^2/\Delta)$ blocks, i.e., $\widehat{v}[[K^{(x)}_i]] \defeq \min \{\widehat{v}[k] \mid k \in K^{(x)}_i\}$. Adding each of these minima to their corresponding values from $\widehat{M}$ gives a list of candidate values for $\widehat{u}[i]$'s, i.e.,
\[\widehat{u}[i] = \min\big\{
0 + \widehat{v}[[K^{(0)}_i]], \:
1 + \widehat{v}[[K^{(1)}_i]], \:
\ldots, \:
\lfloor n / \Delta \rfloor + \widehat{v}[[K^{(\lfloor n / \Delta \rfloor)}_i]]
\big\}.\]
Thus, we already know how to compute $\widehat{u}$ is time $O(n^2 / \Delta)$. Now let us explain how to extend this idea to also list elements of all small enough $C_i$'s. For each value that appears in $\widehat{v}$, the algorithm calculates the sorted sequence of indices under which this value appears in $\widehat{v}$. This allows computing in time $O(\log n)$ how many times a given value appears in a given range of indices in $\widehat{v}$; indeed, it boils down to performing two binary searches of the two endpoints of the range in the sequence corresponding to the given value. Furthermore, all these appearances can be enumerated in time proportional to their count.
For each block $K^{(x)}_i$ such that $x + \widehat{v}[[K^{(x)}_i]] = \widehat{u}[i]$ the algorithm enumerates all appearances of $\widehat{v}[[K^{(x)}_i]]$ and $\widehat{v}[[K^{(x)}_i]]+1$ in the range  $K^{(x)}_i$ in $\widehat{v}$, and adds them to $C_i$. If the total size of $C_i$ would exceed $n / \Delta$, the algorithm stops the enumeration and proceeds to the next block. Similarly, for each block such that $x + \widehat{v}[[K^{(x)}_i]] = \widehat{u}[i] + 1$ the algorithm enumerates all appearances of $\widehat{v}[[K^{(x)}_i]]$.

\proofsubparagraph{(4) Each $\boldsymbol{v}$ is monotone.}
This case is symmetric to the previous one. The difference is that now the algorithm splits $\widehat{v}$ into $O(n / \Delta)$ blocks, and prepares an RMQ data structure for each row of $\widehat{M}$.
\end{proof}

Now we are ready to present our subcubic algorithm for Bounded Monotone Min-Plus OMv, assuming a subcubic algorithm for $\exists$Equality-OMv.

\begin{proof}[Proof of Theorem~\ref{thm:minplus}]
In the preprocessing, the algorithm samples uniformly and independently at random a set $R \subseteq [n]$ of columns of $M$, of size $|R| \defeq \lceil 3 \Delta \ln n \rceil$. For each $r \in R$, the algorithm prepares an $\exists$Equality-OMv data structure for matrix $M^{(r)}$ obtained from $M$ by subtracting the $r$-th column from all the columns, i.e.,
\[M^{(r)}[i,k] \defeq M[i,k] - M[i,r].\]

The algorithm handles each query in two independent steps. The goal of the first step is to compute $u[i]$ for those $i$ that have $|C_i| \leqslant n/\Delta$, and the goal of the second step is to compute $u[i]$ for $i$ with $|C_i| > n/\Delta$.

\proofsubparagraph{First step.}
For each $i \in [n]$, the algorithm either finds out that $|C_i| > n / \Delta$, or lists all elements of $C_i$ and then computes $u[i] = \min_{k \in C_i} (M[i,k] + v[k])$. By Lemma~\ref{lem:listCi}, this takes time $O(n^2 \log n / \Delta)$, for all $i$'s in total. The correctness of this step follows from Lemma~\ref{lem:sufficientCi}.

\proofsubparagraph{Second step.}
In the second step, the algorithm must compute the remaining $u[i]$'s, i.e., those for which $C_i$'s contain too many elements to be handled in the first step. To this end, for every $r \in R$ and every $\delta \in \{0, 1, \ldots, 3 \Delta - 2\}$ the algorithm computes the equality product $M^{(r)} \oeq -(v - v[r] + \delta)$. For every $i \in [n]$, if $(M^{(r)} \oeq -(v - v[r] + \delta))[i]$=1, then there must exist $k \in [n]$ such that
\[M[i,k] - M[i,r] = - (v[k] - v[r] + \delta)\]
and hence
\[M[i,k] + v[k] = M[i,r] + v[r] - \delta.\]
The algorithm therefore adds $M[i,r] + v[r] - \delta$ to the list of possible values for $u[i]$, and at the end of the process it sets each $u[i]$ to the minimum over those values.

\proofsubparagraph{Analysis of the second step.}
We now argue that if $R \cap C_i \neq \emptyset$ (which holds with high probability when $|C_i| > n/\Delta$ via a standard hitting set argument, see below), then the algorithm correctly computes $u[i]$ in the second step. Indeed, pick $r \in R \cap C_i$ and let $k \in [n]$ be a witness for $(M \oplus v)[i]$, i.e., $M[i,k]+v[k]=(M \oplus v)[i]$. Let $\delta \defeq (M[i,r]+v[r]) - (M[i,k]+v[k])$. Clearly, $(M^{(r)} \oeq -(v - v[r] + \delta))[i]$=1, so it only remains to show that $\delta \in \{0,1,\ldots,3\Delta - 2\}$. Obviously, $\delta \geqslant 0$, because $k$ minimizes $M[i,k]+v[k]$. Now let us upper bound the offset $\delta$. Since $r \in C_i$, we have $\widehat{M}[i,r] + \widehat{v}[r] \leqslant \widehat{u}[i]+1$, and hence 
\[M[i,r] + v[r] \ \leqslant \ (\Delta \widehat{M}[i,r] + \Delta - 1) + (\Delta \widehat{v}[r] + \Delta - 1) \ \leqslant \ \Delta \widehat{u}[i] + 3\Delta-2.\]
Moreover, $\widehat{M}[i,k] + \widehat{v}[k] \geqslant \widehat{u}[i]$, and therefore
\[M[i,k] + v[k] \ \geqslant \ \Delta \widehat{M}[i,k] + \Delta \widehat{v}[k] \ \geqslant \ \Delta \widehat{u}[i].\]
We conclude that $\delta \leqslant (\Delta \widehat{u}[i] + 3\Delta-2) - \Delta \widehat{u}[i] = 3\Delta - 2$, as required.
    
It remains to analyze the success probability of the whole algorithm. For a fixed output index $i \in [n]$ such that $|C_i| > n / \Delta$, the probability that the algorithm failed to sample an element $r$ from $C_i$ in all $|R|=\lceil 3\Delta\ln n \rceil$ rounds is at most $(1-1/\Delta)^{3\Delta\ln n} < (1/e)^{3 \ln n} = 1 / n^3$. By a union bound over all $n$~output indices for each of the $n$ queries, the algorithm succeeds to correctly compute all $n^2$ output entries with probability at least $1-n^2/n^3 = 1-1/n$.

\proofsubparagraph{Running time.}
The first step of each query (Lemma~\ref{lem:listCi}) takes time $O(n^2\log n / \Delta)$, summing up to $O(n^3\log n / \Delta)$ for all $n$ queries. Regarding the second step, for each query the algorithm computes $O(|R|\Delta) = O(\Delta^2 \log n)$ equality matrix-vector products, and over the course of $n$ queries this takes time $O(n^{3-\varepsilon} \Delta^2 \log n)$. The total running time is therefore $O(n^3 \log n / \Delta + n^{3 - \varepsilon} \Delta^2 \log n) = O(n^{3-(\varepsilon/3)} \log n)$.
\end{proof}

\section{Remaining reductions}
\label{sec:remaining}

In this section we state the remaining easy reductions that complete the proof of Theorem~\ref{thm:main}.

\subsection{Reduction from \texorpdfstring{$\boldsymbol{\exists}$}{∃}Dominance-OMv to \texorpdfstring{$\boldsymbol{\exists}$}{∃}Equality-OMv}

\begin{observation}\label{obs:dom}
If \itE{}Equality-OMv can be solved in time $T(n)$, then \itE{}Dominance-OMv can be solved in time $O(T(n) \log n)$.
\end{observation}

\begin{proof}
The proof follows a folklore argument, see, e.g., \cite{LabibUW19}. It uses the fact that, for any two non-negative integers $a$ and $b$, it holds that $a < b$ if and only if there exists $\ell \geqslant 0$ such that
\begin{itemize}
\item the $\ell$-th least significant bit of $a$ is $0$; and
\item the $\ell$-th least significant bit of $b$ is $1$; and
\item $a$ agrees with $b$ on all bits higher than the $\ell$-th least significant, i.e., $\lfloor a/2^{\ell + 1} \rfloor = \lfloor b/2^{\ell+1} \rfloor$.
\end{itemize}
Moreover, without loss of generality, all the input numbers are integers between $0$ and $n^2-1$. Indeed, in the preprocessing, each entry of $M$ can be replaced by its rank in the sorted order of all entries of $M$; then, during a query, each entry of $v$ can be replaced by the rank of the smallest entry of $M$ greater than or equal to it. Last but not least, $M[i,k] \leqslant v[k]$ if and only if $M[i,k] < v[k]+1$, because the input numbers are integers. Hence, the algorithm sets, for $\ell = 0, 1, \ldots, \lceil \log (n^2) \rceil$,
\begin{align*}
M^{(\ell)}[i,k] &:=
\begin{cases} 
  \left\lfloor \frac{M[i,k]}{2^{\ell+1}} \right\rfloor,
  & \text{if the $\ell$-th least significant bit of $M[i,k]$ is $0$} \\
  -1, & \text{otherwise},
\end{cases}\\[.5em]
v^{(\ell)}[k] &:=
\begin{cases} 
  \left\lfloor \frac{v[k]+1}{2^{\ell+1}} \right\rfloor,
  & \text{if the $\ell$-th least significant bit of $v[k]+1$ is $1$} \\
  -2, & \text{otherwise},
\end{cases}
\end{align*}
and uses the fact that $(M \ole v)[i] = 1$ if and only if $\exists_\ell \ (M^{(\ell)} \oeq v^{(\ell)})[i] = 1$.
\end{proof}

\subsection{Reduction from Min-Witness-OMv to Min-Max-OMv}

\begin{observation}\label{obs:wit}
If Min-Max-OMv can be solved in time $T(n)$, then Min-Witness-OMv can be solved in time $T(n) + O(n^2)$.
\end{observation}

\begin{proof}
The proof follows another folklore argument, see, e.g., \cite{Lincoln0W20}. The algorithm sets
\[M'[i,k] := \begin{cases}k,&\text{if }M[i,k]=1\\\infty,&\text{otherwise},\end{cases} \quad \text{and} \quad v'[k] := \begin{cases}k,&\text{if }v[k]=1\\\infty,&\text{otherwise},\end{cases}\]
and uses the fact $M \ow v = M' \omm v'$.
\end{proof}

\subsection{Reduction from Boolean-OMv to Bounded Monotone Min-Plus-OMv}

\begin{observation}\label{obs:bool}
If Bounded Monotone Min-Plus-OMv can be solved in time $T(n)$, then Boolean-OMv can be solved in time $T(n) + O(n^2)$.
\end{observation}

\begin{proof}
The algorithm sets
\[M'[i,k] := 2 \cdot (i + k) - M[i,k], \quad \text{and} \quad v_j'[k] := 2 \cdot (j - k) - v_j[k],\]
and uses the fact that $(Mv_j)[i] = 1$ if and only if $(M' \oplus v_j')[i] = 2 \cdot (i + j) - 2$.
\end{proof}
We remark that the above reduction produces Min-Plus-OMv instances that are monotone in all four directions simultaneously, while our Bounded Monotone Min-Plus-OMv algorithm of Theorem~\ref{thm:minplus} works already for instances with monotonicity in one (arbitrarily chosen) direction.

\section*{Acknowledgments}

We thank anonymous reviewers for pointing us to the literature on the cell-probe model and for many useful hints on improving the manuscript.

\bibliographystyle{alphaurl}
\bibliography{main}

\end{document}